\documentclass{article}
\usepackage{authblk}

\usepackage{Definitions}
\usepackage{tcolorbox}

\usepackage{fullpage}
\usepackage[numbers, sort&compress]{natbib}

\usepackage{graphicx} 

\title{On the Hardness of Approximation of the Fair $k$-Center Problem}
\author[1]{Suhas Thejaswi}
\affil[1]{Department of Computer Science, Aalto University, Finland}

\date{}

\begin{document}

\maketitle

\begin{abstract}
In this work, we study the hardness of approximation of the fair $k$-center problem. In this problem, we are given a set of data points in a metric space that is partitioned into groups and the task is to choose a subset of $k$-data points, called centers, such that a prescribed number of data points from each group are chosen while minimizing the maximum distance from any point to its closest center. Although a polynomial-time $3$-approximation is known for fair $k$-center in general metrics, it has remained open whether this approximation guarantee is tight or could be further improved, especially since the classical unconstrained $k$-center problem admits a polynomial-time factor-$2$ approximation. We resolve this open question by proving that, assuming $\p \neq \np$, for any $\epsilon>0$, no polynomial-time algorithm can approximate fair $k$-center to $(3-\epsilon)$-factor.

Our inapproximability results hold even when only two disjoint groups are present and at least one center must be chosen from each group. Further, it extends to the canonical one-per-group setting with $k$-groups (for arbitrary $k$), where exactly one center must be selected from each group. Consequently, the factor-$3$ barrier for fair $k$-center in general metric spaces is inherent, and existing $3$-approximation algorithms are optimal up to lower-order terms even in these restricted regimes. This result stands in sharp contrast to the $k$-supplier formulation, where both the unconstrained and fair variants admit factor-$3$ approximation in polynomial time.
\end{abstract}

\section{Introduction} \label{sec:introduction}

Clustering is a fundamental problem in data analysis and optimization, with applications ranging from uncovering structure in large datasets to selecting representatives for data summarization, and supporting facility-location decisions~\citep{kleindessner2019fair,chen2016matroid,shmoys1994computing,hajiaghayi2012local,hochbaum1985best,jones2020fair}. In its standard form, clustering involves selecting $k$ representative points---called centers---and assigning every data point to a chosen center, thereby inducing a partition of the data into $k$ clusters. The quality of the resulting solution is measured by an objective that quantifies how well the centers ``cover'' the data points.
A particularly well-studied objective is the $k$-center problem, which seeks $k$ centers that minimize the maximum distance from any point to its nearest center. This min-max distance criterion captures a natural ``worst-case'' notion of quality---closely related to an egalitarian fairness principle, since it focuses on protecting the most poorly served point. Though the problem is \np-hard, it admits a simple and practical constant-factor approximation algorithm: starting with a random point as a center and greedily choosing farthest point from the chosen centers (for $k-1$ iterations) achieves a $2$-approximation and this guarantee is essentially tight, assuming $\p \neq \np$~\citep{gonzalez1985clustering,hochbaum1985best}. Owing to its practicality, $k$-center has become a canonical testbed for understanding both the algorithmic possibilities and the inherent limitations of center-based clustering.

When data points are associated with sensitive attributes that induce demographic groups, and clustering is used to select representatives---for example, when the chosen centers are presented as a summary of the data, or used to serve as exemplars for downstream decision-making---optimizing a standard clustering objective alone can yield biased representation for certain groups~\cite{brubach2022fairness,chhabra2021anoverview}. A growing body of work has therefore investigated fair clustering that explicitly encode desiderata related to sensitive attributes via explicit representation constraints~\cite{thejaswi2021diversity, thejaswi2024diversity, thejaswi2022clustering,gadekar2025fair,kleindessner2019fair,jones2020fair,hajiaghayi2012local,krishnaswamy2011matroid}. 
A particularly natural notion for data summarization is representative (or center-based) fairness, which requires the set of chosen centers to include a prescribed numbers of representatives from each group, termed as the fair $k$-center problem~\citep{kleindessner2019fair,jones2020fair}. 
We emphasize that (fair) data summarization is studied in two closely related formulations: $k$-center and $k$-supplier objectives. In the (fair) $k$-center setting, any data point can be chosen as a center and every point is also a client that must be served. In the (fair) $k$-supplier setting, centers must be chosen from a designated subset of points, called suppliers, which may differ from the client set. This distinction is subtle but consequential: in (fair) $k$-supplier, suppliers are not required to be served by the opened centers, whereas in (fair) $k$-center every data point, including potential center locations, must be served~\cite{gadekar2025fair,chen2024approximation,gonzalez1985clustering,hochbaum1986unified}.

On the algorithmic side, fair $k$-center admits constant-factor approximation algorithms in general metrics. The current best-known approximation guarantee is factor-$3$: first compute a ``fairness-oblivious'' set of candidate centers using Gonzalez's greedy procedure achieving factor-$2$ approximation, and then project these candidates to a set of ``fair centers'' that satisfy the group requirements via a matching-based step. The analysis bounds the extra cost incurred during this projection, yielding an overall $3$-approximation~\citep{jones2020fair}.
Strikingly, the picture has remained incomplete. While a factor-$3$ approximation is known for fair $k$-center, it was not known whether this guarantee is tight. This uncertainty is particularly salient when contrasted with the classical (unconstrained) $k$-center problem, which admits a polynomial-time $2$-approximation and for which this factor-$2$ is essentially optimal. It is therefore natural to ask whether fairness constraints inherently force a strictly worse approximation factor, or whether we could hope to ``recover'' the factor-$2$ guarantee in the fair setting. The comparison to the $k$-supplier variant further sharpens the question: both unconstrained $k$-supplier and fair $k$-supplier admit polynomial-time $3$-approximation, and this factor is tight in both cases, so fairness does not create an approximation gap there. In contrast, for $k$-center an apparent gap persists between the $2$-approximation for the unconstrained variant and the $3$-approximation for the fair variant. The main goal of this work is to determine whether this gap can be bridged, \ie, whether the problem has a genuine computational barrier or this is an artifact of current algorithmic techniques employed to solve this problem.

Understanding these limitations are important not only from a theoretical stand point but also in practice. Hardness of approximation results clarify the algorithmic limits of optimization problems such as (fair) clustering: they delineate which improvements are achievable and help avoid chasing approximation guarantees that are provably unattainable. Equally important, they sharpen the research agenda by indicating where progress may still be possible---for example, by exploiting additional structure in metric spaces, \eg, Euclidean or doubling metrics, or by considering meaningful relaxations to fairness constraints. In this paper, we extend this line of research, by investigating the hardness of approximation of the fair $k$-center problem.

\medskip
\xhdr{Technical challenges}
Our hardness of approximation results build upon the inapproximability of the $k$-center problem with costs. Specifically, \citet{chuzoy2005asymmetric}[Theorem 4.1] show that $k$-center with costs is $\np$-hard to approximate within a factor of $3-\epsilon$, for any $\epsilon>0$. In this variant, each point has a nonnegative opening cost, and the total cost of chosen centers should not exceed a given budget.
The hardness of $k$-center with costs carries over to the closely related $k$-center with forbidden centers problem~\cite{angelidakis2021improved}, in which the set of data points $U$ is partitioned into two groups $A \cup B$ and it is allowed to open centers only in $A$, in other words, points in $B$ are forbidden as center locations.\footnote{Note that, in their $(3-\epsilon)$-inapproximability proof for $k$-center with costs, \citet{chuzoy2005asymmetric}[Theorem 4.1] assign unit cost to the ``regular'' points and prohibitively large costs (larger than $k$ suffices) to a designated subset of points. This choice of cost assignment effectively rules out those points with larger costs as possible centers while still requiring them to be covered by an opened center. Consequently, their construction can be interpreted as an instance of $k$-center with forbidden centers: some points must be served, but are not allowed to be selected as centers. \citet{angelidakis2021improved} makes this connection explicit. Citing \citet{chuzoy2005asymmetric}, they note that $k$-center with forbidden centers is $\np$-hard to approximate to a factor $(3-\epsilon)$, and argue that $k$-center with forbidden centers is a special case of the matroid center problem. Consequently, the same $(3-\epsilon)$ inapproximability transfers to matroid center as well. Although their discussion is framed around hardness of approximation for matroid center, the underlying source of the barrier is precisely the $(3-\epsilon)$-inapproximability of $k$-center with forbidden centers (or costs) from \citet{chuzoy2005asymmetric}[Theorem 4.1].}
This formulation is different from $k$-supplier, where the facility set can be distinct from the client set, and suppliers themselves need not be served by any opened center. In contrast, in $k$-center with forbidden centers (equivalently, a ``$k$-supplier'' view with forbidden locations), points in $B$ are disallowed as center locations, but they remain part of the metric and therefore must still be assigned to some opened center in $A$.

Observe that $k$-center with forbidden centers is a special case of fair $k$-center, with two groups $G_1=A$, $G_2=B$ and requirements $r_1=k$, $r_2=0$. This implies an immediate $(3-\epsilon)$-inapproximability in the degenerate regime where some groups can have zero requirement. 
In contrast, many fair clustering settings assume that the data is partitioned into disjoint (demographic) groups $\GG=\{G_1,\dots,G_t\}$ with $G_i\cap G_j=\emptyset$ and $\bigcup_{i\in[t]} G_i = U$; for instance, in the two-group setting one may partition points by sex (\ie, \texttt{Male} and \texttt{Female}). In such scenarios, allowing $r_i=0$ effectively permits selecting no centers from an entire demographic group---an outcome that is usually inconsistent with the intended fairness goal and arguably too weak to justify inapproximability claims for a practically relevant problem, which extends beyond such degenerate cases. Most importantly, the above reduction does not establish inapproximability in the regime where every group must be represented among the centers, \ie, when the requirements are $1 \le r_i \le k-1$.

Moreover, several prior works~\cite{thejaswi2021diversity,thejaswi2022clustering,gadekar2025fair,thejaswi2024diversity,gadekar2025capacitated,cohenaddad2019on,gadekar2025faircommittee} use a common reduction paradigm: they transform fair clustering instances with arbitrary group requirements into a simpler canonical form consisting of exactly $k$ disjoint groups, with the constraint that the solution must select exactly one center from each group. Many fairness constraints and clustering variants admit such a reduction, so an efficient algorithm for this ``one-per-group'' setting would immediately yield efficient algorithms for the more general formulations.
This approach has been particularly powerful for objectives such as $k$-median, $k$-means, and $k$-supplier: after reducing to the one-per-group case, one can obtain approximation guarantees that match those of the corresponding unconstrained problems. This raises a natural question.

\begin{tcolorbox}
\textbf{Question:} \emph{Is fair $k$-center $\np$-hard to approximate within a factor strictly better than $3$ in polynomial time when the fairness constraints are non-degenerate---specifically, when exactly one center must be chosen from each group?}
\end{tcolorbox}

\medskip
\xhdr{Our contributions}
We establish hardness of approximation for fair $k$-center and resolve an open problem, by showing that the factor-$3$ approximation achievable in polynomial time is essentially optimal. Our results cover two non-degenerative cases: ($i$) instances with two groups, where at least one center must be selected from each group, and ($ii$) instances with an arbitrary number of groups in the one-per-group regime, where exactly one center must be selected from each group.

\medskip
The remainder of the paper is organized as follows. In \Cref{sec:problem} we introduce the problem definitions. \Cref{sec:hardness} presents our hardness of approximation results. \Cref{sec:related} reviews the most relevant related work, and \Cref{sec:conclusions} concludes with a discussion of implications of our inapproximability results and possible directions of future work.

\section{Problem Definitions} \label{sec:problem}

Before proceeding further, let us formally define the problems relevant to our work. We begin with the definition of the fair $k$-center problem.

\begin{definition}[The fair $k$-center problem]
Given a metric space $(U,d)$ with $|U|=n$ data points partitioned into $t$ groups $\GG=\{G_i\}_{i \in [t]}$ and a requirement vector $\rvec = \{r_i\}_{i \in [t]}$ such that $\sum_{i \in [t]} r_i = k$. The task is to choose a subset $S \subseteq U$ of $k$ points such that $|S \cap G_i| = r_i$ while minimizing the maximum distance from any point to its closest center in $S$ \ie,
$$\text{minimize} \quad \cost(S, U):= \max_{u \in U} \min_{s \in S} d(u,s).$$
\end{definition}

We also establish hardness of approximation for a variant of fair $k$-center, namely the one-per-group $k$-center problem, which is defined as follows.

\begin{definition}[The one-per-group fair $k$-center problem]
This is a special case of fair $k$-center in which the data points are partitioned into exactly $k$ disjoint groups $\GG=\{G_1,\dots,G_k\}$ and the requirement vector is $\rvec=\vec{1}_k$, \ie, the solution must select exactly one center from each group.
\end{definition}

Our hardness of approximation results are established via a reduction from the $k$-center problem with forbidden centers, which we define next.

\begin{definition}[The $k$-center with forbidden centers problem]
Given a metric space $(U,d)$ whose $n=|U|$ points are partitioned into two sets $A$ and $B$, and an integer $k$, the goal is to choose a set of centers $S \subseteq A$ with $|S|=k$ so as to minimize the maximum distance from any point to its nearest chosen center, \ie,
\[
\text{minimize} \quad \cost(S, U) := \max_{u \in U}\; \min_{s \in S} d(u,s).
\]
\end{definition}

\section{Hardness of Approximation} \label{sec:hardness}

In this section, we present our hardness of approximation results. We first show that fair $k$-center with two groups remains $\np$-hard to approximate within a factor of $3-\epsilon$, for any $\epsilon>0$, even when the requirements are non-degenerate, \ie, when at least one center must be selected from each group. We then extend this hardness to the canonical one-per-group regime by transforming any fair $k$-center with two groups instance into an equivalent instance with $k$ groups with a requirement of choosing exactly one-center per group. The transformation duplicates points to create the desired number of groups, and we preserve the metric structure by introducing a small separation $\delta>0$ between duplicates. By choosing $\delta$ appropriately, we ensure that, for every $\epsilon>0$, a $(3-\epsilon)$-approximation for the one-per-group variant would imply a $(3-\epsilon)$-approximation for the two-group case, and hence such an approximation algorithm cannot exist.

\subsection{Two groups regime}
\begin{theorem}\label{theorem:hardness_twogroups}
Assuming $\p \neq \np$, for any $\epsilon > 0$, there exists no polynomial-time $(3-\epsilon)$-approximation algorithm for the fair $k$-center problem even when there are only two groups and at least one center must be selected from each group.
\end{theorem}
\begin{proof}
We prove the claimed inapproximability result via a polynomial-time reduction from $k$-center with forbidden centers to fair $k$-center with two groups. 
%
Our argument proceeds in three steps.
($i$) We construct, in polynomial time, a fair $k$-center instance with two groups by adding a single auxiliary point $x$ and defining its distances to all other points so that $x$ is effectively isolated (\Cref{claim:reduction}).
($ii$) We prove that every optimal solution to the resulting fair instance must include $x$ as a center (\Cref{claim:optx}), and that removing $x$ preserves the objective value on the remaining points, establishing an equivalence between optimal solutions of the two problems (\Cref{claim:sameopt}).
($iii$) We show that any $\alpha$-approximate solution for the fair $k$-center with two groups instance can be converted into an $\alpha$-approximate solution for $k$-center with forbidden centers instance by simply discarding $x$ (\Cref{claim:alpha_approx}).
Combining these steps with the known $(3-\epsilon)$-inapproximability of $k$-center with forbidden centers yields our hardness of approximation for fair $k$-center with two groups.

\begin{claim}\label{claim:reduction}
    An instance of the $k$-center problem with forbidden centers can be transformed into an equivalent instance of the fair $k$-center problem with two groups in polynomial time.
\end{claim}
\begin{proof}
Let $I=((U,d),A,F,k)$ be an instance of $k$-center with forbidden centers. We construct an instance
$I'=\bigl((U',d'),\GG=\{G_1,G_2\},\rvec=(r_1,r_2),k'\bigr)$
of fair $k$-center with two groups as follows. Introduce a new point $x$ and set
$U' := U \cup \{x\}$, $k' := k+1.$
Define two groups 
$G_1 := A, G_2 := F \cup \{x\}$,
and set the group requirements to
$ r_1 := k, r_2 := 1$.
To define the metric, let
$D := \max_{a,b\in U} d(a,b)$, and set $d'$ on $U'$ by
\begin{align*}
    d'(u,v) &:= d(u,v) &&\text{for all } u,v\in U,\\
    d'(x,x) &:= 0,\\
    d'(u,x) \;=\; d'(x,u) &:= 3D+1 &&\text{for all } u\in U.
\end{align*}
This completes the construction of $I'$. 

Since $(U,d)$ is a metric, $d’$ inherits non-negativity, symmetry, and the triangle inequality for all pairs of points in $U$ (because $d’(u,v)=d(u,v)$ for $u,v\in U$). Moreover, by definition, $d’(x,x)=0$ and for every $u \in U$ we have
$d'(x,u) = d'(u,x) = 3D+1 > 0$,
so non-negativity and symmetry also hold on all pairs involving $x$.
It remains to verify the triangle inequality for triples that include $x$. Fix any $u,v\in U$ (so $u,v\in U’\setminus{x}$). Then
\[
d'(u,x)+d'(x,v) = (3D+1) + (3D+1) = 6D+2.
\]
On the other hand, since $D:=\max_{a,b\in U} d(a,b)$, we have $d'(u,v) = d(u,v) \le D$. Hence,
\begin{align*}
d'(u,v)\;&\le\; D \;<\; 6D+2 \;=\; d'(u,x)+d'(x,v),\\
d'(u,v)\;&\le d'(u,x)+d'(x,v).
\end{align*}
The other two inequalities, $d'(u,x)\le d'(u,v)+d'(v,x)$ and $d'(x,v)\le d'(x,u)+d'(u,v)$, follow immediately since the right-hand side in each case is at least $d'(v,x)=3D+1$ or $d'(x,u)=3D+1$, respectively. Therefore, $d'$ satisfies the triangle inequality on $U'$, and is a metric.

Since $|U'| = |U| + 1$, and all distances can be computed in polynomial time, hence the reduction can be done in polynomial time.
\end{proof}

\begin{claim}\label{claim:optx}
Every optimal solution $S'$ to the fair $k$-center instance must include $x$, \ie, $x\in S'$. Furthermore any optimal solution for fair $k$-center must have cost at most $D$.
\end{claim}
\begin{proof}
    We prove this by contradiction. Suppose that $S'$ is an optimal fair solution with $x\notin S'$. By construction, $d(x,x)=0$, while for every $u\in U\setminus\{x\}$ we have
    $ d(x,u)=3D+1.$
    Since $x$ is not opened as a center, $x$ must be assigned to some center in $S' \subseteq U' \setminus \{x\}$, incurring cost $3D+1$.

    On the other hand, the instance admits a feasible fair solution of cost $D$ (\eg, by selecting any $k$ centers from $G_1$ together with $x$ in $G_2$), and therefore the optimum value is at most $D$. Now replace an arbitrary center $s\in S' \cup G_2 \setminus \{x\}$ by $x$, and consider
    \[
    S'' := S' \setminus \{s\} \cup \{x\}.
    \]
    In $S''$, the point $x$ is served at distance $0$, and the remaining points can be assigned exactly as in $S'$ except possibly those previously assigned to $s$, which can be reassigned without increasing the maximum distance beyond $D$ (since an optimal solution of value at most $D$ exists). Consequently, $S''$ has objective value at most $D$, strictly improving upon the assumed cost of $S'$ (which is at least $3D+1$), contradicting the optimality of $S'$. Therefore, every optimal solution of instance $I'$ must contain $x$.
\end{proof}

\begin{claim}\label{claim:sameopt}
    Any optimal solution for the fair $k$-center instance can be transformed into an optimal solution for the $k$-center with forbidden centers instance and vice versa, without changing the objective function value.
\end{claim}
\begin{proof}
    Let $S^*$ be an optimal solution to fair $k$-center with two groups instance $I'$. By \Cref{claim:optx}, every optimal solution for $I'$ must include $x$, \ie, $x\in S^*$. Moreover, since the group requirement for $G_2$ is $1$, the point $x$ is the unique center selected from $G_2$, and the remaining $k$ centers are selected from $G_1$ (which correspond to the points in $A$).
    Next observe that, any point $u\in U'\setminus\{x\}$ assigned to $x$ would incur distance $d(u,x)=3D+1$. By construction, this value is strictly larger than the distance from $u$ to any feasible center in $G_1$ which is at most $D$. Hence, in an optimal solution, no point other than $x$ itself is assigned to $x$; the center $x$ serves only $x$. It follows that the optimal cost of fair $k$-center is determined entirely by the $k$ centers in $G_1$, and we have
    \[
    \opt(I') = \cost(S^*,U') \;=\; \cost\bigl(S^*\setminus\{x\},\, U'\setminus\{x\}\bigr) = \opt(I).
    \]
    Therefore, letting $\widehat S := S^*\setminus\{x\}$ yields a feasible solution for $k$-center with forbidden centers instance $I$, with the same objective value as $S^*$. Since both instances optimize the same objective over $U'\setminus\{x\}$, the centers $S^*=S^*\setminus \{x\}$ is optimal solution for the $k$-center with forbidden centers instance $I$.
    
    For the other direction, let $\widehat{S}$ be an optimal solution for $k$-center with forbidden centers instance $I'$, and define $S^* := \widehat S\cup\{x\}$. Then $S'$ is feasible for the fair $k$-center instance (it selects exactly one center from $G_2$ and $k$ centers from $G_1$). By \Cref{claim:optx}, any optimal fair solution must contain $x$; omitting $x$ would force $x$ to be served by some other center, which (by the construction) would determine the cost. As argued above, adding $x$ does not change the maximum distance for the remaining points, so $S^*$ optimizes for the same objective as $\Swhat$ and is therefore optimal for the fair $k$-center with two groups instance $I'$. This implies that, $\opt(I) = \opt(I')$.
\end{proof}

\begin{claim}\label{claim:alpha_approx}
Any $\alpha$-approximation algorithm for fair $k$-center with two groups instance yields an $\alpha$-approximation algorithm for $k$-center with forbidden centers instance (via the reduction described in \Cref{claim:reduction}).
\end{claim}
\begin{proof}
Let $S'$ be the set of centers returned by an $\alpha$-approximation algorithm for the constructed fair instance $I'$. By \Cref{claim:optx}, every feasible (and hence every approximate) solution to $I'$ must include $x$, so $x\in S$. Define $S := S' \setminus \{x\}$.
Then $S \subseteq A$ and $|S|=k$, so $S$ is feasible for $k$-center with forbidden-centers instance $I$.
By \Cref{claim:sameopt}, the objective values coincide after removing $x$, namely
\[
\cost(S',U') \;=\; \cost(S,\,U'\setminus\{x\}) \;=\; \cost(S,U),
\]
and moreover $\opt(I)=\opt(I')$. Since $S'$ is an $\alpha$-approximate solution for $I'$,
\[
\cost(S',U') = \cost(S,I) \;\le\; \alpha\cdot \opt(I').
\]
Combining the above inequality with \Cref{claim:optx}, which states $\opt(I') = \opt(I)$, we obtain,
\[
\cost(S,U) \;\le\; \alpha\cdot \opt(I),
\]
so $S$ is an $\alpha$-approximation for the $k$-center problem with forbidden centers.
\end{proof}

For the $k$-center problem with forbidden centers, it is known that, unless $\p=\np$, no polynomial-time $(3-\epsilon)$-approximation algorithm exists, for any $\epsilon>0$. Together with \Cref{claim:alpha_approx}, this implies that a $(3-\epsilon)$-approximation for fair $k$-center would immediately yield a $(3-\epsilon)$-approximation for $k$-center with forbidden centers, contradicting the known hardness of approximation result of $k$-center with forbidden centers. Therefore, unless $\p=\np$, fair $k$-center with two groups cannot be approximated within a factor strictly better than $3$ in polynomial time, and hence the known factor-$3$ approximation guarantee is tight.
\end{proof}

\begin{remark}
Note that the inapproximability result of \Cref{theorem:hardness_twogroups} extends to fair $k$-center with two groups with an arbitrary (non-zero) requirement for each group. The only modification in the construction of \Cref{claim:reduction} is to introduce $r_2 \in \{1,\dots, k-1\}$ auxiliary points $\{x_1,\dots,x_{r_2}\}$ and place them in group $G_2$ (\ie, $G_2=B \cup \{x_1,\dots,x_{r_2}\}$), $U' = U \cup \{x_1,\dots,x_{r_2}\}$, and then set $r_1=k$, and $k' = k + r_2$. Furthermore, set distances in $d'$ as follows: 
\begin{align*}
d'(x_i,x_i) &=0, \quad \text{for all}\quad x_i \in \{x_1,\dots,x_{r_2}\},\\
d'(u,v) &=d(u,v), \quad \text{for all} \quad u,v \in U,\\
d'(x_i,u) = d'(u,x_i) &= 3D+1, \quad \text{for each}\quad x_i \in \{x_1,\dots,x_{r_2}\},\; \text{and}\quad u \in U' \setminus \{x_i\}.
\end{align*}
The analysis then carries over: any feasible (and hence any optimal) solution must open points $\{x_1,\dots,x_{r_2}\}$ as centers in $G_2$, and they do not affect the cost of serving the points in $U$. Consequently, the $(3-\epsilon)$-inapproximability in polynomial time holds for fair $k$-center with two groups with requirements $1 \le r_1 \le k-1$ and $1 \le r_2 \le k-1$.
\end{remark}

\begin{remark}
We present our hardness of approximation result in the two-group regime to highlight that the $(3-\epsilon)$-inapproximability already holds when the number of groups is small. Moreover, such two-group instances can be transformed into equivalent instances with more groups (up to $k$) by duplicating points and adjusting the per-group requirements accordingly. Hence, the same $(3-\epsilon)$-inapproximability extends to settings with an arbitrary number of groups; in this sense, the two-group regime is a strictly restrictive case. This transformation using duplication step is exactly what we use to establish the hardness result for the one-per-group fair $k$-center problem, which we explain in the subsection to follow.
\end{remark}

\subsection{One-per-group regime}

Next we will establish the hardness of approximation for the one-per-group fair $k$-center problem. Our result is stated in the following theorem.

\begin{theorem}
Assuming $\p \neq \np$, for any $\epsilon > 0$, there exists no polynomial-time $(3-\epsilon)$-approximation algorithm for the fair $k$-center problem even when exactly one center must be selected from each group.
\end{theorem}
\begin{proof}
We prove the claimed inapproximability via a polynomial-time reduction from the fair $k$-center problem with two groups (already shown to be inapproximable to ($3-\epsilon$)-factor in \Cref{theorem:hardness_twogroups}) to an instance of one-per-group fair $k$-center. Concretely, we show that any instance with two disjoint groups and an arbitrary (non-degenerate) requirements $1 \leq r_1 \leq k-1$ and $1 \le r_2 \le k-1$ can be transformed, in polynomial time, into an equivalent instance of fair $k$-center consisting of exactly $k$ disjoint groups with the constraint that exactly one center must be chosen from each group. The reduction duplicates points to create the desired number of groups, and we preserve the metric structure by introducing a small separation $\delta > 0$ between duplicates. Our argument proceeds in four steps: ($i$) we construct the one-per-group fair $k$-center instance (\Cref{claim:one_per_group_reduction}), ($ii$) we prove that the constructed instance form a metric space (\Cref{claim:one_per_group_metric}), ($iii$) we show that the optimum objective value is preserved (\Cref{claim:one_per_group_sameopt}), and ($iv$) we show that any $\alpha$-approximate solution to the one-per-group instance yields an $\alpha$-approximate solution to the fair $k$-center with two groups instance (\Cref{claim:one_per_group_alphaapx}). Combining these results with \Cref{theorem:hardness_twogroups} yields the claimed hardness of approximation result.

\begin{claim}\label{claim:one_per_group_reduction}
Any fair $k$-center with two groups instance can be reduced, in polynomial time, into an equivalent instance of one-per-group fair $k$-center.
\end{claim}
\begin{proof}
Let
$I = \bigl( (U,d), \; \GG=\{G_1,G_2\}, \; \rvec=(r_1,r_2), \; k \bigr)$ 
be an instance of fair $k$-center with two disjoint groups $G_1$, $G_2$ such that $r_1+r_2$ = $k$ and $1 \le r_1,r_2 \le k-1$. We construct a one-per-group fair $k$-center instance
$I' = \bigl((U',d'), \; \HH, \; \vec{1}_k, \; k \bigr)$.
in which the points are partitioned into $k$ disjoint groups and exactly one center must be chosen from each group.
For our hardness result, it would already suffice to apply this transformation to the particular two-groups instance with $r_1=k$ and $r_2=1$ constructed in \Cref{claim:reduction}. Nevertheless, we describe the reduction for fair $k$-center with two groups with arbitrary requirements $1 \leq r_1,r_2 \leq k-1$, as it may be of independent interest. Our construction is as follows.
\squishlist
\item \textit{$k$ disjoint groups.} Create $r_1$ disjoint copies of $G_1$ denoted
$\{H^1_{1}$, $H^1_{2}$, \ldots, $H^1_{r_1}\}$, and $r_2$ disjoint copies of $G_2$ denoted
$\{H^2_{1}, H^2_{2}, \ldots, H^2_{r_2}\}$. Let $\HH=\{H^1_1,\dots,H^1_{r_1},H^2_{1},\dots,H^2_{r_2}\}$ be the family consisting of these $k = r_1+r_2$ groups. The requirement vector is $\rvec = \vec 1_{k}$, \ie, exactly one center must be selected from each group in $\HH$.

\item \textit{Point duplication.} For each $u \in G_1$ and each $a \in [r_1]$, create a distinct copy $u^1_{a}$ and place it in $H^1_{a}$. Similarly, for each $v \in G_2$ and each $b \in [r_2]$, create a distinct copy $v^2_{b}$ and place it in $H^2_{b}$. Let $U'=\bigcup_{H \in \HH} H$ denote the union of all such copies. By construction, the groups in $\HH$ are pairwise disjoint.

\item \textit{Distances.} Let $\Delta := \min\{ d(u,v) : u,v \in U,\; u \neq v\}$,
and set $\delta := \smallfrac{\Delta}{2}$. Define $d'$ on $U'$ as follows. For any two points $x,y \in U'$, let $\pi(x),\pi(y) \in U$ denote their corresponding original points in $U$.
\[
d'(x,y) :=
\begin{cases}
0 & \text{if } x=y,\\
\delta & \text{if } x \neq y \text{ and } \pi(x)=\pi(y),\\
d(\pi(x),\pi(y)) & \text{if } \pi(x) \neq \pi(y).
\end{cases}
\]
\squishend

This completes the construction of $I'$. Since $|U'| \leq k \cdot |U|$ and the distances are computable in polynomial time, hence the reduction is polynomial.
\end{proof}

\begin{claim}\label{claim:one_per_group_metric}
The distance function $d'$ constructed in \Cref{claim:one_per_group_reduction} satisfy the metric properties on $U'$.
\end{claim}
\begin{proof}
Symmetry holds because $d$ is a metric and it holds from the construction that $d'(x,y)=d'(y,x)$ for all $x,y \in U'$. Non-negativity is immediate, since $d'(x,y)=0$ iff $x=y$, and $d'(x,y)\geq \delta$ and $\delta>0$ for all $x \neq y, x,y \in U'$. It remains to verify the triangle inequality. Fix any $x,y,z \in U'$ and let $a=\pi(x)$, $b=\pi(y)$, $c=\pi(z)$ be their original points in $U$. We consider the following cases:
\squishlist
\item[($i$)] If $a,b,c$ are all distinct, then $d'(x,y)=d(a,b)$, $d'(y,z)=d(b,c)$, and $d'(x,z)=d(a,c)$. So the triangle inequality of $(U',d')$ follows directly from the triangle inequality in $(U,d)$.

\item[($ii$)]  If $a=b=c$, then all non-zero pairwise distances among $x,y,z$ equal $\delta$, and it holds that
\begin{align*}
\delta  + \delta &> \delta, \\
d'(x,y) + d'(y,z) &> d'(x,z).
\end{align*}

\item[($iii$)] If exactly two points in $a,b,c$ are the same. Without loss of generality assume that $a=c \neq b$, which implies $d(a,c)=0$ and $d(a,b)=d(c,b) \geq \Delta$. Furthermore, $d'(x,y)=d'(z,y)=d(a,b)=d(c,b)$. The first triangle inequality holds as follows,
%
%
\begin{align*}
d'(x,y) + d'(y,z) &= 2 \cdot d(a,b) \ge 2 \cdot \Delta > \Delta > \delta = d'(x,z), \\
d'(x,y)+d'(y,z) &> d'(x,z).
\end{align*}
The next triangle inequality holds as follows.
\begin{align*}
d'(x,z) + d'(z,y) &= \delta + d(c,b) \\
&= \delta + d(a,b) \\
&= \delta + d'(x,y) \\
d'(x,z) + d'(z,y) &> d'(x,y).
\end{align*}
The final triangle inequality holds because of a similar argument to the previous.
\begin{align*}
d'(y,x) + d'(x,z) &= d(b,a) + \delta \\
&= d(b,c) + \delta \\
&= d'(y,z) + \delta \\
d'(y,x) + d'(x,z) &> d'(y,z).
\end{align*}
\squishend

Therefore $d'$ satisfies the triangle inequality in all cases. 
Since $d'$ satisfy all properties of a metric on points in $U'$, we conclude that $(U',d')$ is a metric space.
\end{proof}

\begin{claim}\label{claim:one_per_group_sameopt}
An optimal solution of fair $k$-center with groups instance $I$ can be transformed into an optimal solution for one-per-group fair $k$-center instance $I'$ in polynomial time and vice versa. More strongly, the cost of optimal solutions is the same in both instances.
\end{claim}
\begin{proof}
 Let $\opt(I)$ and $\opt(I')$ denote the optimal solution of instances $I$ and $I'$, respectively. We show both inequalities will hold, \ie, $\opt(I') \leq \opt(I)$ and $\opt(I) \leq \opt(I')$, thus establishing the equality.

\squishlist
\item[($i$)] $\opt(I') \leq \opt(I)$. Let $S^\star$ be an optimal solution to $I$, so $|S^\star \cap G_i| = r_i$ for $i \in \{1,2\}$. Construct a feasible solution $\Swhat$ for $I'$ by assigning each chosen center $s \in S^\star \cap G_1$ to a distinct group $H^1_{a}$ and selecting the corresponding copy $s^1_{a}$; do the same for centers in $S^\star \cap G_2$ across the groups $H^2_{b}$. This selects exactly one center from each group in $\HH$, so $\Swhat$ is a feasible solution for the one-per-group fair $k$-center instance $I'$.

Now consider any client point $x \in U'$ with original point $u=\pi(x)$ and hence $u \in U$. Let $s \in S^\star$ be a closest center to $u \in U$. By construction, $\Swhat$ contains some copy $s'$ of $s$. If $u \neq s$, then
$d'(x,s') = d(u,s) \le \mathrm{cost}_{I}(S^\star)$. If $u=s$, then either $x=s'$ (distance $0$) or $x$ is a different copy of the same original point (distance $\delta$); in either
case $d'(x,s') \le \max\{\cost_{I}(S^\star),\delta\}$. Since $S^*$ is an optimal solution for $I$, we have
\begin{equation} \label{eq:one_per_group:max_cost}
\cost_{I'}(\Swhat) \le \max\{ \cost_{I}(S^\star), \delta\} = \max\{\opt(I),\delta\}.
\end{equation}
Finally, observe that finding an optimal solution for instances with $|U| \leq k$ is trivial, since choosing all points in $U$ suffices. So the \np-hard instance must be $k < |U|$. In this case, every feasible solution in $I$ has cost at least $\Delta > \delta$, so the optimal solution $\opt(I) \ge \Delta > \delta$. Hence
\begin{equation} \label{eq:one_per_group:opt_cost}
\max\{\cost_I(S^*),\delta)\} = \cost_I(S^*) \quad \text{and}\quad \max\{\opt(I),\delta\} = \opt(I),
\end{equation}
combining \Cref{eq:one_per_group:max_cost} with \Cref{eq:one_per_group:opt_cost} and the fact that $\opt(I') \leq \cost_{I'}(\Swhat)$ we obtain
\[
\opt(I') \le \cost_{I'}(\Swhat) \le \max\{\opt(I),\delta\} = \opt(I).
\]

\item[($ii$)] $\opt(I) \le \opt(I')$. Let $\Swhat$ be an optimal solution to $I'$, \ie, it chooses exactly one center from each group in $\HH$. Consider the multiset of original points $\pi(\Swhat) \subseteq U$ counted with multiplicity. Since $I'$ has exactly $r_i$ groups derived from $G_i$, the multiset $\pi(\Swhat)$ contains exactly $r_i$ elements from $G_i$ for each $i \in \{1,2\}$, counting multiplicity. Let $\Swcheck$ be the set of distinct original points appearing in $\pi(\Swhat)$. If $|\Swcheck \cap G_i| < r_i$ for some $i \in \{1,2\}$ due to duplicates, add arbitrary points from $G_i \setminus \Swcheck$ until $|\Swcheck \cap G_i|=r_i$. This yields a feasible solution for $I$, and adding centers does not increase the cost.

For a fixed original point $u \in U$, and choose any duplicate point $x \in U'$ with $\pi(x)=u$. Let $R := \cost_{I'}(\Swhat)$, by definition of $R$ there exists a center $\swhat \in \Swhat$ such that $d'(x,\swhat) \le R$. Let $\swcheck=\pi(y)$ be the original point corresponding to the duplicate center $\swhat$, so $\swcheck$ must be in the distinct centers $\Swcheck$ without duplicity, \ie, $\swcheck \in \Swcheck$.

If $u \neq \swhat$ then $d'(x,\swhat)=d(u,\swcheck)$, hence $d(u,\swcheck) \le R$. If $u=\swhat$ then $d(u,\swhat)=0 \le R$. Therefore every $u \in U$ is within
distance $R$ of $\Swcheck$, \ie,
\begin{equation}\label{eq:sameopt_2}
\mathrm{cost}_{I}(\Swcheck) \le R = \mathrm{cost}_{I'}(\Swhat) = \opt(I').
\end{equation}
Taking the minimum over feasible $\Swcheck$ in $I$ and considering the fact that $\opt(I) \leq \cost_I(\Swcheck)$, gives $\opt(I) \le \opt(I')$.
\squishend

Finally, by combining ($i$) and ($ii$) we conclude $\opt(I')=\opt(I)$.
\end{proof}

\begin{claim}\label{claim:one_per_group_alphaapx}
Any $\alpha$-approximation algorithm for the one-per-group fair $k$-center instance $I'$ yields an $\alpha$-approximation algorithm for the fair $k$-center with two groups instance $I$ (via the reduction described in \cref{claim:one_per_group_reduction}).
\end{claim}
\begin{proof}
Let $\Swhat$ be the set of centers returned by an $\alpha$-approximation algorithm for $I'$, so
\[
\cost_{I'}(\Swhat) \le \alpha \cdot \opt(I').
\]
Convert $\Swhat$ into a feasible solution $\Swcheck$ for $I$ using the procedure described in \Cref{claim:one_per_group_sameopt} point ($ii$), \ie, take $\Swcheck$ to be the set of distinct originals points in $\pi(\Swhat)$ and, if necessary, add arbitrary points within each group $G_i$ to satisfy $|\Swcheck \cap G_i|=r_i$. As shown in \Cref{eq:sameopt_2}, this conversion satisfies
\[
\cost_{I}(\Swcheck) \le \cost_{I'}(\Swhat).
\]
Therefore,
\begin{align*}
\cost_{I}(\Swcheck) &\le \cost_{I'}(\Swhat) \le \alpha \cdot \opt(I') = \alpha \cdot \opt(I),\\
\cost_{I}(\Swcheck) & \leq \alpha \cdot \opt(I),
\end{align*}
where the equality follows from \Cref{claim:one_per_group_sameopt} that establishes $\opt(I) = \opt(I')$. Hence $\Swhat$ is an $\alpha$-approximate solution to the fair $k$-center with two groups instance $I$.
\end{proof}

To conclude, for the sake of contradiction let us assume that there exists a polynomial-time $(3-\epsilon)$-approximation algorithm for the one-per-group fair $k$-center problem. By \Cref{claim:one_per_group_alphaapx}, this would imply a polynomial-time $(3-\epsilon)$-approximation algorithm for the fair $k$-center problem with two groups, contradicting \Cref{theorem:hardness_twogroups}. Therefore, unless $\p=\np$, no polynomial-time $(3-\epsilon)$-approximation algorithm exists for the one-per-group fair $k$-center problem.
\end{proof}

\section{Further Related Work} \label{sec:related}

Our work builds upon broad literature on (fair) clustering and the computational complexity of clustering problems. Since a comprehensive review is beyond the scope of our work, we focus on the results most closely related to our contributions and refer an interested reader to surveys on (fair) clustering~\citep{jain1999data, chhabra2021anoverview}.

Clustering is a fundamental problem in computer science, which has been extensively studied in both from theoretical and applied viewpoint~\cite{jain1988algorithms,vazirani2001approximation}. Among many clustering objectives, two well known center-based clustering formulations are $k$-center and $k$-supplier, both of which aim to minimize the maximum distance from each client to its nearest center. The key distinction lies in the set of admissible centers: in $k$-center, any data point can be chosen as a center, whereas, in $k$-supplier centers must be chosen from a designated subset of data points, referred as suppliers~\cite{gonzalez1985clustering,hochbaum1985best,hochbaum1986unified,shmoys1994computing,goyal2023tight}. Both problems are $\np$-hard~\citep{vazirani2001approximation}, yet they admit polynomial-time constant-factor approximation algorithms: $k$-center has a $2$-approximation~\citep[Theorem 2.2]{gonzalez1985clustering}, and $k$-supplier has a $3$-approximation~\citep[Theorem 5]{hochbaum1986unified}. Moreover, these guarantees are essentially best possible under standard complexity theory assumptions: unless $\p=\np$, for any $\epsilon>0$, there exists no polynomial-time $(2-\epsilon)$-approximation algorithm for $k$-center~\citep[Theorem 4.3]{gonzalez1985clustering} and no polynomial-time $(3-\epsilon)$-approximation algorithm for $k$-supplier~\citep[Theorem~6]{hochbaum1986unified}.
Furthermore, many variants of $k$-center (and $k$-supplier) remain $\np$-hard, and their inapproximability barriers often persist even under additional constraints. A variant most closely related to our work is $k$-center with costs, for which \citet{chuzoy2005asymmetric}[Theorem 4.1] prove $(3-\epsilon)$ inapproximability for every $\epsilon>0$, assuming $\p \neq \np$. This hardness result further extends to the closely related $k$-center problem with forbidden centers, where a subset of points are forbidden to be chosen as centers~\cite{angelidakis2021improved}. Our hardness of approximation results are established via a reduction from $k$-center with forbidden centers.

Algorithmic fairness has gained considerable attention in light of mounting evidence that optimization and learning systems, when designed without accounting for sensitive attributes and structural disparities in the data, can amplify or even create discriminatory outcomes~\cite{kleinberg2018algorithmic,baeza2018bias}. This has motivated a broader effort to revisit classical combinatorial optimization problems through the lens of fairness, by incorporating additional constraints that in turn require new algorithmic ideas across a range of settings~\cite{chierichetti2017fair,samadi2018price,matakos2024fair,kleindessner2019fair,gadekar2025capacitated,chen2024approximation,gadekar2025faircommittee,zehlike2022fairness}.
Within this broader agenda, fairness in clustering has emerged as a particularly active area, driven by the role of clustering in unsupervised machine learning and data summarization. Many notions of fairness have been proposed---capturing proportional representation within clusters or explicit constraints on which representatives are selected---leading to a rich algorithmic literature~\cite{ghadiri2021socially,hajiaghayi2012local,krishnaswamy2011matroid,gadekar2025capacitated}. In this work, we focus on cluster-center fairness: points are associated with demographic attributes forming groups and fairness constraints are imposed by restricting the number of centers chosen from each group, while optimizing standard objectives such as $k$-median, $k$-means, $k$-center, and $k$-supplier.
Existing formulations capture several types of center-selection constraints, including exact requirements~\citep{jones2020fair,kleindessner2019fair}, lower-bound constraints~\citep{thejaswi2021diversity,thejaswi2022clustering}, upper-bound constraints~\citep{hajiaghayi2012local,krishnaswamy2011matroid,chen2016matroid}, and combined upper-and-lower bound constraints~\citep{hotegni2023approximation,thejaswi2024diversity,gadekar2025capacitated}. Our focus in this paper is the fair $k$-center problem with exact constraints.

\citet{kleindessner2019fair} introduced the fair $k$-center problem in the context of data summarization, and gave a $(3\cdot 2^{t-1})$-approximation in $\bigO(nkt^2 + kt^4)$ time. This was was later improved to a factor-$3$ approximation with running time $\bigO(nk + n\sqrt{k}\log k)$ by \citet{jones2020fair}. More recently, \citet{gadekar2025fair} gave a tight $3$-approximation for fair $k$-supplier in $\bigO((kn + k^2 \sqrt{k})\log n \log k)$ time, matching the approximation factor of unconstrained $k$-supplier. All these results assume a setting in which the groups form a disjoint partition on the data points.
When groups are allowed to overlap, the problem becomes substantially more difficult: in general, even deciding feasibility (\ie, whether there exists any size-$k$ subset of data points that satisfies the fairness requirements for all groups) is $\np$-hard, which rules out any meaningful approximation guarantees for the clustering objective. Consequently, algorithmic work on intersecting groups typically focuses on parameterized regimes, for example, assuming that the number of centers and/or the number of groups is small or otherwise bounded, under which feasibility and optimization becomes tractable~\cite{gadekar2025fair,thejaswi2022clustering,thejaswi2024diversity,zhang2024parameterized}. In a recent work, \citet{gadekar2025capacitated} showed that for certain fair clustering formulations with intersecting groups, inapproximability can persist even when feasibility is trivial (\ie, when satisfying the fairness requirements is decidable in polynomial time). While their results are stated for fair range $k$-median (and $k$-means), it suggests a broader phenomenon that may also arise for fair $k$-center and fair $k$-supplier.

To the best of our knowledge, prior to our work there were no hardness of approximation results known for the fair $k$-center problem with disjoint groups, particularly in the non-degenerate regime where at least one center must be chosen from each group.

\section{Conclusions} \label{sec:conclusions}

In this work, we presented a comprehensive analysis of the computational complexity of the fair $k$-center problem in general metric spaces. We show that the known factor-$3$ approximation guarantee is essentially optimal: under standard complexity assumption $\p \neq \np$, it cannot be improved. Consequently, any better approximation ratio must rely on additional structure beyond arbitrary metrics---for example, by restricting attention to Euclidean or doubling metrics, where geometric properties can be exploited.
Our findings stand in sharp contrast with the closely related $k$-supplier objective: in that setting, both the fair and the unconstrained problem variants admit polynomial-time factor-$3$ approximation algorithms, which is not the case for $k$-center.

\medskip
\xhdr{Acknowledgments} We thank Aristides Gionis, Bruno Ordozgoiti, and Ameet Gadekar for insightful discussions on the computational complexity of fair clustering over the past several years, which have been an important inspiration for this work.

\bibliographystyle{ACM-Reference-Format}
\bibliography{refs}

\end{document}